\documentclass[]{IEEEtran2}
\usepackage{amsmath,amssymb,amsthm}
\usepackage{graphicx,color,enumerate}
\usepackage{algpseudocode}
\usepackage{algorithm}
\DeclareMathOperator*{\argmax}{arg\,max}
\newtheorem{theorem}{Theorem}[section]

\newtheorem{mydef}{Definition}

\newcommand{\fis}[3]{\frac{1+\sum_{#2\neq #1}|h_{#1#2}|^2 #3_{#2}(h)}{|h_{#1#1}|^2}}

\title{Algorithms for Stochastic Games on Interference Channels}
\author{\IEEEauthorblockN{Krishna Chaitanya A, Utpal Mukherji, Vinod Sharma} \\
\IEEEauthorblockA{Department of ECE, Indian Institute of Science, Bangalore-560012 \\ Email: $\lbrace$akc, utpal, vinod$\rbrace$ @ece.iisc.ernet.in}
}
\date{}
\begin{document}
\maketitle
\begin{abstract}
We consider a wireless channel shared by multiple transmitter-receiver pairs. Their transmissions interfere with each other.  Each transmitter-receiver pair aims to maximize its long-term average transmission rate subject to an average power constraint.  This scenario is modeled as a stochastic game.  We provide sufficient conditions for existence and uniqueness of a Nash equilibrium (NE).  We then formulate the problem of finding NE as a variational inequality (VI) problem and present an algorithm to solve the VI using regularization.  We also provide distributed algorithms to compute Pareto optimal solutions for the proposed game.
\end{abstract}
\section{Introduction}
\hspace{0.5cm} We consider a wireless channel which is being shared by multiple users to transmit their data to their respective receivers.  The transmissions of different users may cause interference to other receivers.  This is a typical scenario in many wireless networks.  In particular, this can represent inter-cell interference on a particular wireless channel in a cellular network.  The different users want to maximize their transmission rates.  This system can be modeled in the game theoretic framework and has been widely studied \cite{palomar} - \cite{spectrum}.\par
\hspace{0.5cm} In \cite{palomar}, the authors have considered parallel Gaussian interference channels.  This setup is modeled as a strategic form game and existence and uniqueness of a Nash equilibrium (NE) is studied.  The authors provide conditions under which the water-filling function is a contraction and thus obtain conditions for uniqueness of NE and for convergence of iterative water-filling.  They extend these results to a multi-antenna system in \cite{MIMO_IWF} and consider an asynchronous version of iterative water-filling in \cite{async}. \par
\hspace{0.5cm} Parallel Gaussian interference channels (PGIC) were also treated in \cite{stochastic}, \cite{pareto}, \cite{VI}, \cite{luca}, \cite{pang_qos}. In \cite{stochastic}, authors describe an online algorithm to find NE. \cite{VI} proposes a variational inequality approach to choose a NE when there exist multiple NE. A 2-user PGIC is considered in \cite{luca} and proposes a channel selection game and finds its NE. In \cite{pareto} and \cite{pang_qos}, authors consider minimizing power consumption subject to quality of service (QoS) constraints. In \cite{pang_qos}, authors find NE and in \cite{pareto}, when strategy space is limited to finite power levels, Pareto optimal points are proposed as a solution of the game.  In \cite{spectrum}, authors formulate the problem of interference channels as a Stackelberg game and study its equilibrium.  We consider power allocation in a non-game-theoretic framework in \cite{PA} (see also other references in \cite{PA} for such a setup).  In \cite{PA}, we have proposed a centralized algorithm for finding the Pareto points that maximize sum rate.\par
\hspace{0.5cm} All the above cited works consider a one shot non-cooperative game (or a Pareto point).  As against that we consider a stochastic game over Gaussian interference channels, where the users want to maximize their long term average rate and have long term average power constraints (for potential advantages of this over one shot optimization, see \cite{goldsmith}, \cite{vsharma}).  For this system we obtain existence of NE and also develop algorithms to obtain NE via variational inequalities and using regularization.  The convergence of these algorithms is proved under weaker conditions than would be obtained via the methods of \cite{palomar}.  Finally, we provide distributed algorithms to obtain local Pareto points and show their convergence under complete generality.\par
\hspace{0.5cm} The paper is organized as follows. In Section \ref{sys_model}, we present the system model and formulate it as a stochastic game.  In Section \ref{one}, we study this stochastic game and define the basic terminology.  In Section \ref{ne_vi}, we formulate the NE problem as a variational inequality problem and present algorithms to solve the variational inequality.  In Section \ref{pareto}, we discuss the Pareto optimal solutions to the proposed game.  In Section \ref{ne}, we present numerical examples and Section \ref{concl} concludes the paper.
\section{System model and Notation} \label{sys_model}
\hspace{0.5cm} We consider a Gaussian wireless channel being shared by $N$ transmitter-receiver pairs. The time axis is slotted and all users' slots are synchronized.  The channel gains of each transmit-receive pair are constant during a slot and change independently from slot to slot.\par
\hspace{0.5cm} Let $H_{ij}(k)$ be the channel gain from transmitter $j$ to receiver $i$ (for transmitter $i$, receiver $i$ is the intended receiver).  We assume that, $\{H_{ij}(k), k \geq 0 \}$ is an $i.i.d$ sequence with distribution $\pi_{ij}$.  We also assume that these sequences are independent of each other.  The direct channel power gains $\vert H_{ii}(k)\vert^2 \in \mathcal{H}_d = \lbrace h_1,h_2,\dots,h_{n_1} \rbrace$ and the cross channel power gains $\vert H_{ij}(k)\vert^2 \in \mathcal{H}_c = \lbrace g_1,g_2,\dots,g_{n_2}\rbrace$.  We denote $(H_{ij}(k), i,j = 1,\dots,N)$ by ${\bf H}(k)$ and its realization vector by $h(k)$ which takes values in $\mathcal{H}$, the set of all possible channel states.  The distribution of ${\bf H}(k)$ is denoted by $\pi$.  If user $i$ uses power $P_i({\bf H})$ then it gets rate $\text{log} \left(1+\Gamma_i \left(P \left({\bf H} \right) \right) \right)$, where
\begin{equation}
\Gamma_i(P({\bf H})) = \frac{\alpha_i |H_{ii}|^2 P_i({\bf H})}{1 + \sum_{j \neq i}|H_{ij}|^2P_j({\bf H})},
\end{equation}
${\bf H}$ is the channel state vector, $P({\bf H}) = (P_1({\bf H}),\dots,P_N({\bf H}))$ and $\alpha_i$ is a constant that depends on the modulation and coding used by transmitter $i$.  The aim of each user $i$  is to choose a power policy to maximize its long term average rate
\begin{equation*}
r_i({\bf P}_i,{\bf P}_{-i}) \triangleq \limsup\limits_{n \rightarrow \infty} \frac{1}{n} \sum_{k=1}^n \mathbb{E}[\text{log} \left(1+\Gamma_i \left(P \left({\bf H}(k) \right) \right) \right)],
\end{equation*}
subject to average power constraint
\begin{equation*}
\limsup\limits_{n \rightarrow \infty} \frac{1}{n} \sum_{k=1}^n \mathbb{E}[P_i(k)] \leq \overline{P}_i, \text{ for each } i,
\end{equation*}
where ${\bf P}_{-i}$ denotes the power policies of all users except user $i$.\par
\hspace{0.5cm} We address this problem as a stochastic game problem with the set of feasible power policies of user $i$ denoted by $\mathcal{A}_i$ and its utility by $r_i$.  Let $\mathcal{A} = \Pi_{i=1}^{N} \mathcal{A}_i$.\par
\hspace{0.5cm} We limit ourselves to stationary policies, i.e., the power policy for every user in slot $k$ depends only on the channel state $H(k)$ and not on $k$.  In the current setup, it does not entail any loss in optimality.  In fact now we can rewrite this optimization problem to find policy $P({\bf H})$ such that $r_i = \mathbb{E}_{\bf H}[\text{log} \left(1+\Gamma_i \left(P \left({\bf H}\right) \right) \right)]$ is maximized subject to $\mathbb{E}_{\bf H}\left[P_i({\bf H})\right] \leq \overline{P}_i$ for all $i$.  We express power policy of player $i$ by ${\bf P}_i = (P_i(h), h\in \mathcal{H})$, where transmitter $i$ transmits in channel state $h$ with power $P_i(h)$.  We denote the power profile of all players by ${\bf P} = ({\bf P}_1,\dots,{\bf P}_N)$.
\section{Game Theoretic Formulation} \label{one}
\begin{mydef}
A point ${\bf P}^*$ is a Nash Equilibrium (NE) of game $\mathcal{G} = \big((\mathcal{A}_i)_{i=1}^N, (r_i)_{i=1}^N\big)$ if for each player $i$
\begin{equation*}
r_i({\bf P}_i^*,{\bf P}_{-i}^*) \geq r_i({\bf P}_i,{\bf P}_{-i}^*) \text{ for all } {\bf P}_i \in \mathcal{A}_i.
\end{equation*}
\end{mydef}
\hspace{0.5cm} Existence of a pure NE for the strategic game $\mathcal{G}$ follows from the Debreu-Glicksberg-Fan Theorem (\cite{BASAR}, page no. 69), since in our game $r_i({\bf P}_i,{\bf P}_{-i})$ is a continuous function in the profile of strategies ${\bf P} = ({\bf P}_i,{\bf P}_{-i}) \in \mathcal{A}$ and concave in ${\bf P}_i$.
\begin{mydef}
  The best-response of player $i$ is a function $BR_i : \mathcal{A}_{-i} \rightarrow \mathcal{A}_i$ such that $BR_i({\bf P}_{-i})$ is a solution of the optimization problem of maximizing $r_i({\bf P}_i, {\bf P}_{-i})$, subject to ${\bf P}_i \in \mathcal{A}_i$.   
\end{mydef}
\hspace{0.5cm} We see that the Nash equilibrium is a fixed point of the best-response function.  In our game, given the power profile of the other players ${\bf  P}_{-i}$, the best response of player $i$ is
\begin{equation}
BR_i({\bf P}_{-i};h) = \text{max}\left\{0, \lambda_i({\bf P}_{-i}) - \frac{(1+\sum_{j\neq i}\vert h_{ij}\vert^2 P_j(h))}{\vert h_{ii}\vert^2}\right\}, \label{wf}
\end{equation}
where $\lambda_i({\bf P}_{-i})$ is chosen such that the average power constraint is satisfied and $\alpha_i = 1$ for all $i$.  It is easy to observe that the best-response of player $i$ to a given strategy of other players is water-filling on ${\bf f}_i({\bf P}_{-i}) = (f_i({\bf P}_{-i};h),h \in \mathcal{H})$ where 
\begin{equation}
f_i({\bf P}_{-i};h) = \frac{(1+\sum_{j\neq i}\vert h_{ij}\vert^2 P_j(h))}{\vert h_{ii}\vert^2}.
\end{equation}
For this reason, we represent the best-response of player $i$ by ${\bf WF}_i({\bf P}_{-i})$.  The notation used for the overall best-response is same as that used for power profiles, ${\bf WF}({\bf P}) = ({\bf WF}(P(h)), h \in \mathcal{H})$, where ${\bf WF}(P(h)) = (WF_1({\bf P}_{-1};h),\dots,WF_N({\bf P}_{-N};h))$ and $WF_i({\bf P}_{-i};h)$ is as defined in (\ref{wf}).  We use ${\bf WF}_i({\bf P}_{-i}) = (WF_i({\bf P}_{-i};h), h \in \mathcal{H})$.\par 
\hspace{0.5cm} It is observed in \cite{palomar} that the best-response ${\bf WF}_i({\bf P}_{-i})$ is also the solution of the optimization problem
\begin{equation}
\text{ minimize } \left\Vert {\bf P}_i  + {\bf f}_i({\bf P}_{-i})\right\Vert^2, \text{ subject to } {\bf P}_i \in \mathcal{A}_i. \label{opt_proj}
\end{equation}
As a result we can interpret the best-response as projection of $(-f_{i,1}({\bf P}_{-i}),\dots,-f_{i,N}({\bf P}_{-i}))$ on to $\mathcal{A}_i$.  We denote the projection of $x$ on to $\mathcal{A}_i$ by $\Pi_{\mathcal{A}_i}(x)$.  We define the cost function of player $i$, $C_i({\bf P}_i,{\bf P}_{-i}) = \left\Vert {\bf P}_i  + {\bf f}_i({\bf P}_{-i})\right\Vert^2$.  We consider (\ref{opt_proj}), as a game in which every player minimizes its cost function with strategy set of player $i$ being $\mathcal{A}_i$.  We denote this game by $\mathcal{G}^{\prime}$.  This game has the same set of NEs as $\mathcal{G}$ because the best responses of these two games are equal.\par
\hspace{0.5cm} We can rewrite the optimization problem (\ref{opt_proj}) as :
\begin{eqnarray}
&\text{ minimize } &\sum_{h \in \mathcal{H}} \left( P_i(h) + \fis{i}{j}{P} \right)^2, \nonumber \\ 
&\text{ subject to } &{\bf P}_i \in \mathcal{A}_i.\label{proj2}
\end{eqnarray}
We note that this is a convex optimization problem. Necessary and sufficient optimality conditions for a convex optimization problem (\cite{Comp}, page 210) applied to the optimization problem (\ref{proj2}) simplifies to 
\begin{multline}
\sum_{h \in \mathcal{H}} \left( WF_i({\bf P}_{-i};h) + \fis{i}{j}{P} \right) \\
\left( V_i(h) - WF_i({\bf P}_{-i};h) \right) \geq 0,\label{ineq}
\end{multline}
for all ${\bf V}_{i} \in \mathcal{A}_i$.  We can rewrite the $N$ inequalities in (\ref{ineq}) in compact form as 
\begin{equation}
\left({\bf WF}({\bf P}) + \hat{h} + \hat{H}{\bf P} \right)^T \left(x - {\bf WF}({\bf P})\right) \geq 0 \text{ for all } x \in \mathcal{A}, \label{cond_1}
\end{equation}
where $\hat{h}$ is a $N_1$-length block vector with $N_1 = \vert \mathcal{H}\vert$, and each block $\hat{h}(h), h \in \mathcal{H}$, is of length $N$ and is defined by $\hat{h}(h) = \left( \frac{1}{\vert h_{11} \vert^2}, \dots, \frac{1}{\vert h_{NN} \vert^2}\right)$ and  $\hat{H}$ is the block diagonal matrix $\hat{H} = \text{diag}\left\lbrace \hat{H}(h), h \in \mathcal{H} \right\rbrace$ with each block $\hat{H}(h)$ defined by
\begin{equation*}
  [\hat{H}(h)]_{ij} = \begin{cases}
    0 & \text{ if } i=j, \\
    \frac{\vert h_{ij} \vert^2}{\vert h_{ii} \vert^2}, & \text{ else. }
    \end{cases}
\end{equation*}
\hspace{0.5cm} To find a NE, we need to find the fixed points of the waterfilling function for which we use the characterization (\ref{cond_1}).\par
\hspace{0.5cm} A condition for uniqueness of the NE, and for convergence of iterative water-filling for parallel Gaussian interference channels to the NE, was presented in \cite{palomar}.  This condition in the current setup is given by $\rho(S^{max}) < 1$, where the elements of matrix $S^{max}$ are
\begin{equation*}
[S^{max}]_{ij} = \begin{cases}
  0 & \text{ if } i=j, \\
  \text{ max }_{h \in \mathcal{H}} \frac{\vert h_{ij} \vert^2}{\vert h_{ii} \vert^2}, & \text{ else. }
  \end{cases}
\end{equation*}
We study this condition further.
\begin{theorem}
$\rho(S^{max}) < 1$  if and only if 
\begin{equation}
\frac{max\{g_1,\dots,g_{n_2}\}}{min\{h_1,\dots,h_{n_1}\}} < \frac{1}{N-1}.
\label{cond_contr}
\end{equation}
\end{theorem}
\begin{proof}
It can be seen that, all row sums of $S^{max}$ are equal to 
\begin{equation*}
(N-1)\frac{max\{g_1,\dots,g_{n_2}\}}{min\{h_1,\dots,h_{n_1}\}}.
\end{equation*}
Thus, from the Frobenius theorem on spectral radius (\cite{minc}, pp. 24-26), $\rho(S^{max}) < 1$ if and only if inequality (\ref{cond_contr}) holds.
\end{proof}
We need the following result in the next section.
\begin{theorem}
$\rho(\hat{H}) < 1 \text{ if and only if } \rho(S^{max}) < 1.$
\end{theorem}
\begin{proof}
Since the matrix $\hat{H}$ is a block diagonal matrix, $\rho(\hat{H}) \leq 1 \text{ if and only if } \rho(\hat{H}(h)) \leq 1 \text{ for } h \in \mathcal{H}$.  It should be noted that $S^{max}$ is also a diagonal block of the block diagonal matrix $\hat{H}$. Maximum row sum of $\hat{H}(h)$ is upper bounded by that of $S^{max}$. Because, $S^{max} = \hat{H}(h)$ for some $h \in \mathcal{H}$, maximum row sum of $\hat{H}$ is the row sum of $S^{max}$.  Using Frobenius theorem, $\rho(\hat{H}) < 1 \text{ iff } \rho(S^{max}) < 1.$
\end{proof}
\hspace{0.5cm} Therefore under (\ref{cond_contr}), we obtain a unique NE for our problem and iterative water-filling converges to the unique NE.  However, (\ref{cond_contr}) is a strong condition.  In the next section we obtain a weaker condition via variational inequalities.
\section{NE using Variational Inequalities} \label{ne_vi}
\hspace{0.5cm} Theory of variational inequalities offers various learning techniques to find NE of a given game.  The equivalence of finding a NE and solving a $VI$ is noted in \cite{Pang}.  A variational inequality problem denoted by $VI(K,F)$ is defined as follows.
\begin{mydef}
Consider a closed and convex set $K \subset \mathbb{R}^n$, and a function $F:K \to K$. The variational inequality problem $VI(K,F)$ is defined as the problem of finding $x \in K$ such that $$F(x)^T(y-x) \geq 0 \text{ for all } y \in K.$$
\end{mydef}
\begin{mydef}
We say that $VI(K,F)$ is  
\begin{itemize}
\item Monotone if $(F(x) -F(y))^T(x-y) \geq 0 \text{ for all } x,y \in K.$
\item Strictly monotone if $(F(x) -F(y))^T(x-y) > 0 \text{ for all } x,y \in K, x \neq y.$
\item Strongly monotone if there exists an $\epsilon >0 $ such that $(F(x) -F(y))^T(x-y) \geq \epsilon \Vert x-y \Vert^2 \text{ for all } x,y \in K$.
\end{itemize}
\end{mydef}
\hspace{0.5cm} We use the projection algorithm (\cite{Pang}, section 12.1) 
\begin{equation}\label{bpa}
x(l+1) = \Pi_K\left( x(l) - \tau F(x(l)) \right), \text{ for } l = 1,2,\dots,
\end{equation}
to solve strongly monotone $VI(K,F)$.  Convergence of the projection algorithm is proved for strongly monotone variational inequality.  For that, first we formulate our problem as a strongly monotone $VI$ when $\tilde{H}$ is positive semidefinite.\par
\hspace{0.5cm} Consider the variational inequality problem $VI(\mathcal{A},F({\bf P}))$ to find ${\bf P}$ such that,
\begin{equation}
\left(F({\bf P}) \right)^T \left(x - {\bf P}\right) \geq 0 \text{ for all } x \in \mathcal{A}, \label{VI_prob}
\end{equation}
where $F({\bf P}) = \hat{h} + \tilde{H}{\bf P} \text{ and } \tilde{H} = I + \hat{H}.$  The solution ${\bf P}^*$ of (\ref{VI_prob}) is a Nash equilibrium of the game $\mathcal{G}$ as it is a Nash equilibrium of $\mathcal{G}^{\prime}$.\par
\hspace{0.5cm} To use (\ref{bpa}), we first convert $VI(\mathcal{A},F({\bf P}))$ to a strongly monotone $VI$.  Define $F_{\epsilon_n}({\bf P}) = \tilde{H}{\bf P} + \hat{h} + \epsilon_n{\bf P},$ for $\epsilon_n > 0$.  We find conditions for $VI(\mathcal{A},F_{\epsilon_n})$ to be strongly monotone.  Then, using (\ref{bpa}), we can find a solution of $VI(\mathcal{A},F_{\epsilon_n})$.  It is shown in \cite{Pang} that as $\epsilon_n \rightarrow 0$, the solution of $VI(\mathcal{A},F_{\epsilon_n})$ converges to that of $VI(\mathcal{A},F)$.
\begin{theorem}
If $\tilde{H}$ is positive semidefinite, $VI(\mathcal{A},F_{\epsilon_n})$ is a strongly monotone $VI$, for $\epsilon_n > 0$.
\end{theorem}
\begin{proof}
$\left(F_{\epsilon_n}({\bf P}) - F_{\epsilon_n}({\bf V})\right)^T({\bf P} - {\bf V})$
\begin{eqnarray*}
& = & (\tilde{H}{\bf P} + \epsilon_n{\bf P} - \tilde{H}{\bf V} - \epsilon_n{\bf V})^T({\bf P} - {\bf V}) \\
& = & ({\bf P} - {\bf V})^T\tilde{H}^T({\bf P} - {\bf V})+ \epsilon_n({\bf P} - {\bf V})^T({\bf P} - {\bf V}) \\
& \geq & \epsilon_n \Vert ({\bf P} - {\bf V}) \Vert^2.
\end{eqnarray*}
Thus, $VI(\mathcal{A},F_{\epsilon_n})$ is a strongly monotone $VI$.
\end{proof}
\hspace{0.5cm} Thus, we can apply (\ref{bpa}) to solve $VI(\mathcal{A},F_{\epsilon_n})$ for sufficiently small $\epsilon_n > 0$, to get a close approximation of a NE whenever $\tilde{H}$ is positive semidefinite. \par
\hspace{0.5cm} If $\tilde{H}$ is positive definite, $VI(\mathcal{A}, F)$ is a strictly monotone $VI$. A strictly monotone $VI$ admits atmost one solution (\cite{Pang}, page 156).  Since existence of a solution of $VI(\mathcal{A}, F)$ follows from existence of a NE of our game, when $\tilde{H}$ is positive definite this solution is infact unique.
\begin{theorem}
If $\rho(\hat{H}) < 1$ then $\tilde{H}$ is positive definite matrix.
\end{theorem}
\begin{proof}
If $\rho(\hat{H}) < 1$, then all eigenvalues $\lambda_1,\dots,\lambda_N$ of $\hat{H}$ are in unit circle.  Thus, the eigenvalues $1+\lambda_i, i = 1,\dots,N$ of $\tilde{H} = I + \hat{H}$ have positive real parts and hence $\tilde{H}$ is positive definite.
\end{proof}
\hspace{0.5cm} The other-way implication is not true.  For example, consider 3-user interference channel with $\mathcal{H}_d = \{0.3, 0.6\}$ and $\mathcal{H}_c = \{0.2, 0.1\}$. It can be seen that $\rho(\hat{H}) > 1$ but $\tilde{H}$ is positive definite. Thus we can find the NE using (\ref{bpa}).\par
\hspace{0.5cm}  The condition that $\tilde{H}$ is positive semidefinite is a much weaker condition than $\rho(\hat{H}) < 1$. The former condition requires the eigenvalues of $\tilde{H}$ to lie in the right half plane, the latter requires the eigenvalues to lie in a unit circle with $(1,0)$ as center.
\section{Pareto Optimal Solutions}\label{pareto}
\hspace{0.5cm} In this section, we consider Pareto optimal solutions to the game $\mathcal{G}$.  A power allocation ${\bf P}^{*}$ is Pareto optimal if there does not exist a power allocation ${\bf P}$ such that $r_i({\bf P}_i,{\bf P}_{-i}) \geq r_i({\bf P}_i^{*},{\bf P}_{-i}^{*})$ for all $i = 1,\dots,N$ with atleast one strict inequality.  It is well-known that the solution of a weighted-sum optimization of the utility functions is Pareto optimal, i.e., the solution of the following optimization problem, 
\begin{equation}
\text{ max } \sum_{i=1}^N w_i r_i({\bf P}_i,{\bf P}_{-i}),\text{ such that } {\bf P}_i \in \mathcal{A}_i \text{ for all } i, \label{ws_pareto}
\end{equation}
with $w_i > 0$, is Pareto optimal.  Thus, since $\mathcal{A}$ is compact and $r_i$ are continuous, a Pareto point exists for our problem.  We apply the weighted-sum optimization (\ref{ws_pareto}) to the game $\mathcal{G}$ to find a Pareto-optimal power allocation. \par
\hspace{0.5cm} To solve the non-convex optimization problem in a distributed way, we employ augmented Lagrangian method and solve for the stationary points using the algorithm in \cite{distributed}.  We present the resulting algorithm to find the Pareto power allocation in Algorithm \ref{algo_pareto}.  Define the augmented Lagrangian as 
\vspace{-0.5cm}
\begin{multline*}
\mathcal{L}({\bf P},{\bf \lambda}) =  \sum_{i=1}^N w_i r_i({\bf P}_i,{\bf P}_{-i}) + \sum_{i=1}^N \lambda_i(\overline{P}_i-\sum_{h \in \mathcal{H}}\pi(h) P_i(h))\\ + c \sum_{i}(\overline{P}_i-\sum_{h \in \mathcal{H}}\pi(h) P_i(h))^2.
\end{multline*}
\begin{algorithm}
  \begin{algorithmic}
    \State Initialize $\lambda_i^{(1)},{\bf P}_i^{(0)}$ for all $i = 1,\dots,N$.
    \For {$n = 1 \to \infty $}
    \State $ {\bf P}^{(n)} = $ $\text{Steepest\_Ascent}(\lambda^{(n)},{\bf P}^{(n-1)})$\vspace{0.2cm}
    \If {$|\overline{P}_i-\sum_{h}\pi(h) P_i^{(n)}(h)| < \epsilon \text{ for all } i = 1,\dots,N $}
    \State break
    \Else
    \State $\lambda_i^{(n+1)} = \lambda_i^{(n)} - \alpha (\overline{P}_i-\sum_{h}\pi(h) P_i^{(n)}(h))$    
    \State $n = n+1$
    \EndIf
    \EndFor    
    \Function{$\text{Steepest\_Ascent}$}{$\lambda,{\bf P}$}
    \State Fix $\delta, \epsilon$
    \State Initialize $t = 1, {\bf P}^{(t)} = {\bf P}$.
    \Loop
    \For {$i=1 \to N$}
    \State player $i$ updates his power variables as \\ ${\bf Q}_i = {\bf P}_i^{(t)} + \delta \triangledown_i \mathcal{L}({\bf P}_i^{(t)},{\bf P}_{-i}^{(t)},\lambda)$
    \EndFor
    \State Choose ${\bf P}^{(t+1)}$ as 
    \State $i^* = \argmax_i \mathcal{L}({\bf Q}_i,{\bf P}_{-i}^{(t)},\lambda) - \mathcal{L}({\bf P}_i^{(t)},{\bf P}_{-i}^{(t)},\lambda)$ 
    \State ${\bf P}^{(t+1)} = ({\bf Q}_{i^*},{\bf P}_{-i^*}^{(t)}) $
    \State $t = t+1$.
    \State Till $\Vert \triangledown_i\mathcal{L}({\bf P}_i^{(t)},{\bf P}_{-i}^{(t)},\lambda)\Vert_2  < \epsilon$ for each $i$.
    \EndLoop
    \State return ${\bf P}^{(t)}$
    \EndFunction
  \end{algorithmic}
  \caption{Augmented Lagrangian method to find Pareto optimal Power allocation}
  \label{algo_pareto}
\end{algorithm}
We denote the gradient of $\mathcal{L}({\bf P},{\bf \lambda})$ with respect to power variables of player $i$ by $\triangledown_i \mathcal{L}({\bf P}_i,{\bf P}_{-i},\lambda)$.  In Algorithm \ref{algo_pareto}, the step sizes $\alpha,\delta$ are chosen sufficiently small.  Convergence of the steepest ascent function in Algorithm \ref{algo_pareto} is proved in \cite{distributed}. \par
\hspace{0.5cm} Since this is a nonconvex optimization problem, Algorithm \ref{algo_pareto} converges to a local Pareto point (\cite{MOP}) depending on the initial power allocation.  We can get better local Pareto points by initializing the algorithm from different power allocations and choosing the Pareto point which gives the best sum rate among the ones obtained.  We consider this in our illustrative examples.
\section{Numerical Examples}\label{ne}
\hspace{0.5cm} In this section we compare the sum rate achieved at a Nash equilibrium and a Pareto optimal point obtained by the algorithms provided above.  We choose a 3-user interference channel. For Example 1: $\mathcal{H}_d = \{3, 1.5 \}$ and $\mathcal{H}_c = \{0.1, 0.5\}$ and for Example 2: $\mathcal{H}_d = \{0.3, 1\}$ and $\mathcal{H}_c = \{0.2, 0.1\}$.  Here, we assume that all elements of $\mathcal{H}_d, \mathcal{H}_c$ occur with equal probability, i.e., with probability 0.5.  In Example 1, $\rho(\hat{H}) = 0.6667$, hence water-filling function is a contraction and iterative water-filling converges to the unique NE. In Example 2, $\rho(\hat{H}) = 1.3333,$ but $\tilde{H}$ is a positive definite matrix as each block matrix of the diagonal is positive definite.  Thus it has a unique NE.  In Example 2, iterative water-filling does not converge but we can use the regularization algorithm to find the NE. To find Pareto optimal points, in both examples, we choose weights equal to 1 and we use Algorithm \ref{algo_pareto}.  We initialize Algorithm \ref{algo_pareto} from $10$ different initial power allocations chosen at random.  The best Pareto point among the $10$ Pareto points is chosen and plotted in Figure \ref{plot}.  We compare the sum rates for the NE and the Pareto point in Figure \ref{plot} for Example 1 and in Figure \ref{plot2} for Example 2.  In Figures \ref{plot}, \ref{plot2}, we also compare the sum rate at the Pareto point achieved using the algorithm presented in \cite{PA} which is a centralized algorithm and decodes the strong and very strong interference instead of treating them as noise.  The two Pareto optimal curves in Figures \ref{plot}, \ref{plot2} almost coincide, since in both examples all the channel states have weak interference alone, and this interference is treated as noise.  We notice here that Pareto optimal points are more efficient in terms of sum rate than NE.
\begin{figure}
  \hspace{-0.5cm}
  \includegraphics[height=3.9cm,width=9cm]{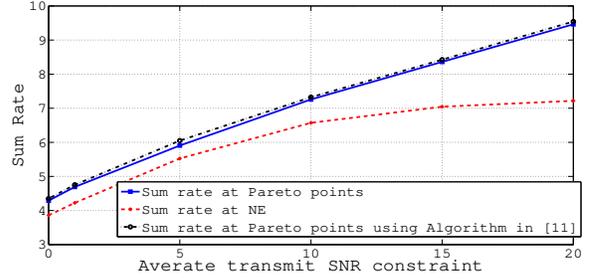}
  \vspace{-0.3cm}
  \caption{Sum rate comparison at Pareto optimal points and Nash equilibrium points for Example 1.}
  \label{plot}    
  \vspace{-0.4cm}
\end{figure}
\begin{figure}
  \hspace{-0.5cm}
  \includegraphics[height=4cm,width=9cm]{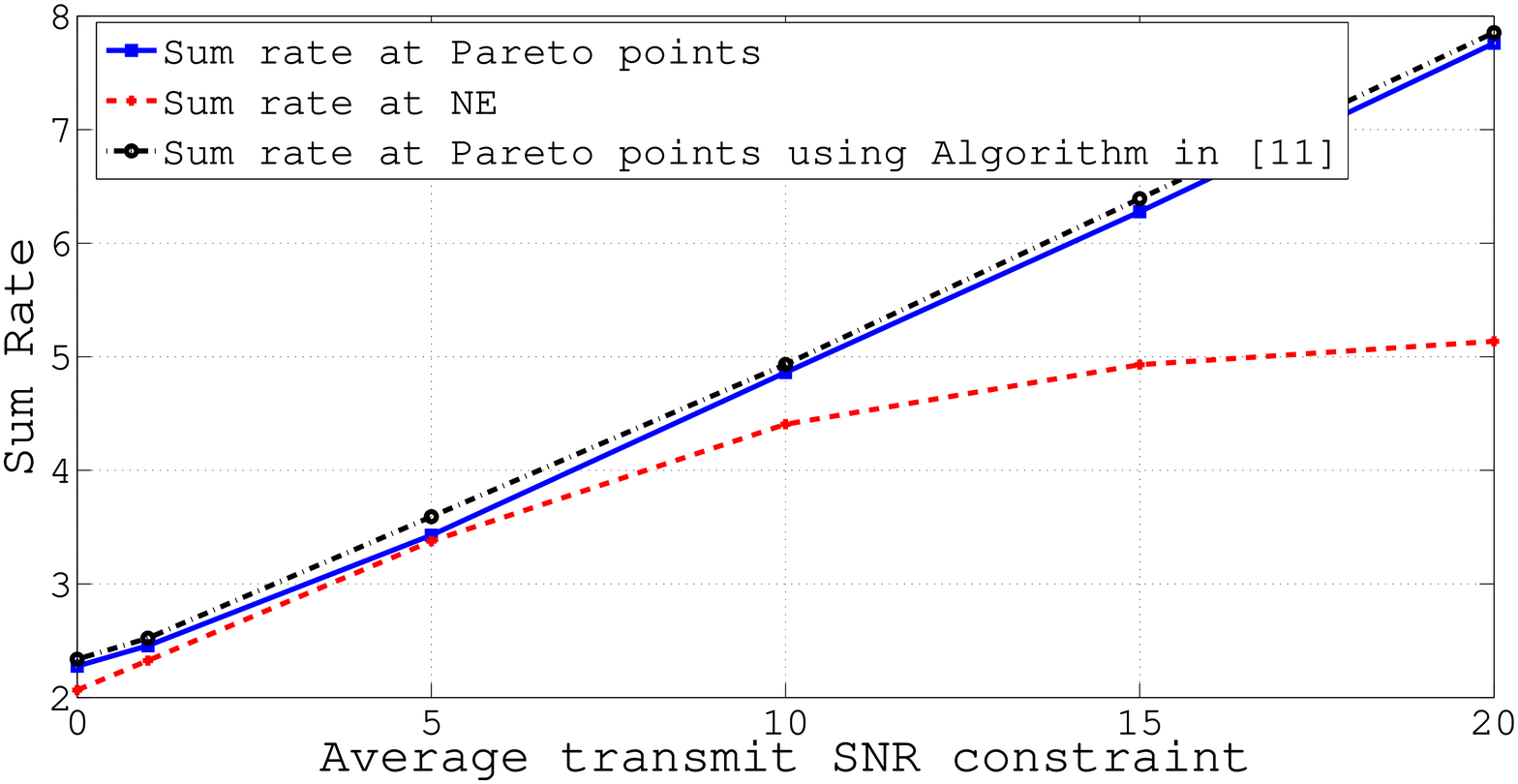}
  \vspace{-0.3cm}
  \caption{Sum rate comparison at Pareto optimal points and Nash equilibrium points for Example 2.}  
  \label{plot2}
  \vspace{-0.5cm}
\end{figure}
\section{Conclusions}\label{concl}
\hspace{0.5cm} We have considered a channel shared by multiple users.  We presented a variational inequality approach using regularization to find the NE of the proposed power allocation game.  The conditions required for convergence of the algorithm based on VI are weaker than those of iterative water-filling.  We have also presented a distributed algorithm to find local Pareto optimal solutions.  This algorithm converges under general conditions and provides more efficient solutions than the NE.

\end{document}